\documentclass{article}
\usepackage{amssymb,amsmath}
\usepackage{graphicx,enumerate}

\title{A portfolio choice problem in the framework of expected utility operators}

\author{Irina Georgescu \\ \footnotesize Academy of Economic Studies\\ \footnotesize Department of Economic Informatics and Cybernetics\\ \footnotesize Pia$\c{t}$a Romana No 6  R 70167, Oficiul Postal 22, Bucharest, Romania\\
 \footnotesize Email: irina.georgescu@csie.ase.ro \\ and \\ Louis Aim\'{e} Fono\footnote{A part of this work has been done under the research Grant No $17$-$497RG/MATHS/AF/AC_G-FR3240297728$ offered by The World Academy of Sciences (TWAS) to the Applied Mathematics to Social Sciences Research Group of the Laboratory of Mathematics-University of Douala-Cameroon. The second author (member of the Group) sincerely thanks TWAS.} \\  \footnotesize
Laboratory of  Mathematics and Department of Mathematics and Computer Sciences,\\ \footnotesize Faculty of Sciences, University of  Douala,
P.O. Box 24157, Douala, Cameroon\\
\footnotesize Email: lfono2000@yahoo.fr}

\date{}

\begin{document}
\maketitle

\begin{abstract}
Possibilistic risk theory starts from the hypothesis that risk is modelled by fuzzy numbers. In particular, in a possibilistic portfolio choice problem, the return of a risky asset will be a fuzzy number.

The expected utility operators have been introduced in a previous paper to build an abstract theory of possibilistic risk aversion. To each expected utility operator one can associate a notion of possibilistic expected utility. Using this notion, we will formulate in this very general context a possibilistic choice problem. The main results of the paper are two approximate calculation formulas for corresponding optimization problem. The first formula approximates the optimal allocation with respect to risk aversion and investor's prudence, as well as the first three possibilistic moments. Besides these parameters, in the second formula the temperance index of the utility function and the fourth possibilistic moment appear.
\end{abstract}

\textbf{Keywords}: Expected utility operators; possibilistic moments; portfolio choice problem

\newtheorem{definitie}{Definition}[section]
\newtheorem{propozitie}[definitie]{Proposition}
\newtheorem{remarca}[definitie]{Remark}
\newtheorem{exemplu}[definitie]{Example}
\newtheorem{intrebare}[definitie]{Open question}
\newtheorem{lema}[definitie]{Lemma}
\newtheorem{teorema}[definitie]{Theorem}
\newtheorem{corolar}[definitie]{Corollary}

\newenvironment{proof}{\noindent\textbf{Proof.}}{\hfill\rule{2mm}{2mm}\vspace*{5mm}}

\section{Introduction}

An agent investing a wealth $w_0$ in a risk--free asset and a risky asset is interested in determining the optimal proportion in which to divide $w_0$ between the two assets. The mathematical model is a decision problem (called the standard portfolio choice problem) formulated in the context of von Neumann-Morgenstern expected utility (\cite{arrow}, \cite{brandt}, \cite{eeckhoudt2}, \cite{gollier}, \cite{pratt}). The agent will choose the optimal wealth allocation in the risky asset by determining the solution of an optimization problem. An exact solution of the optimization problem is difficult to obtain, therefore approximate solutions have been searched for. Using Taylor-type approximations, various forms of approximate solutions have been found with respect to two classes of parameters: indicators of the investor's risk performance (absolute risk aversion, prudence, temperance, etc.) and various moments of the return of the risky asset (see \cite{athayde}, \cite{courtois}, \cite{garlappi}, \cite{niguez}, \cite{zakamulin}).

The investment moments in the mentioned papers are probabilistic and the return of the risky asset being a random variable. In real life, on can consider risky asset as a fuzzy number.  In paper \cite{georgescu5}, a possibilistic porfolio problem was formulated using the notion of possibilistic expected utility from \cite{georgescu2} based in the  framework of Zadeh' s possibility theory \cite{zadeh}. We determined an approximate calculation formula of the optimal solution in terms of the first three posibilistic moments of the fuzzy numbers representing the return of the risky asset and the investor's risk aversion and prudence indicators.

On the other hand, in papers \cite{georgescu1} or \cite{georgescu4}, there has been proposed the notion of possibilistic expected utility different from the one from \cite{georgescu2}. Using this notion, we can formulate another possibilistic portfolio choice problem, with a different solution from \cite{georgescu2}. The expected utility operators have been introduced in \cite{georgescu3} to elaborate a general theory of possibilistic risk aversion. Each expected utility operator defines a notion of possibilistic expected utility. By particularization, we obtain the possibilistic expected utilities of \cite{georgescu1}, \cite{georgescu2}. In addition, Tassak et al. \cite{tassak1} and Sadefo et al. \cite{sadefo} have developped an approach based on credibity theory which do not consider in this paper.

In this paper we will formulate a standard portfolio choice problem in the context of a remarkable class of expected utility operators: $D$--operators. A first result of the paper is an approximate calculation formula of the optimal solution depending on the risk aversion and prudence indicators, as well as the first three moments defined by the considered $D$-operator. The second result is a more refined approximation formula of the optimal solution: besides the above mentioned parameters, in the solution component appear the investor's temperance and the fourth order $T$-moment. In particular, forms of the optimal solution from the problems associated with the $D$-operators from \cite{georgescu1}, \cite{georgescu2} are  obtained.

In Section 2, by analogy with the notion of (probabilistic) expected utility from classical von Neumann-Morgenstern utility theory,  two different notions of possibilistic expected utilities are recalled.
Definitions of expected utility operators and the two  usual  examples are also recalled. By a natural property of derivability,  the $D$-operators are introduced, offering a general framework for the study of possibilistic portfolio choice problem.\\
Section 3 states our portfolio choice problem, proposes $T$-standard model ($T$ is a D-operator) and  solutions of the model. More precisely,  the  model is formulated and it is analyzed the way to find two approximate optimal solutions for the optimization problem associated with such a model. The first one is in terms of risk aversion and prudence of the utility function of the investor and the first three $T$-moments associated with a fuzzy number $A$ in the component of the excess return (in case of a small portfolio risk). The  second approximate formula  is based on the above mentioned indicators and on the investor's temperance and the fourth-order moment.
% This formula reflects more deeply the way different parameters influence the optimal allocation of the $T$-standard model.\\
Section 4 contains some concluding remarks and formalizes some open problems.

\section{Review on Possibilistic expected utility theories and introduction of D-operators}
\subsection{Possibilistic expected utility and possibilistic variance of a fuzzy number}

The von Neumann-Morgenstern expected utility theory ($EU$-theory), a natural framework to study risk parameters phenomena,  is based on two elements (\cite{arrow}, \cite{eeckhoudt2}, \cite{gollier}):

$\bullet$ a utility function $u$ representing an agent;

$\bullet$ a random variable $X$ representing the risk.\\

Thus,  the (probabilistic) expected utility $E(u(X)),$  defined as the mean value of the random variable $u(X),$   is the fundamental notion of $EU$-theory providing  the possibilistic indicators associated with $X$ such as expected value, variance, moments, covariance, etc. In addition, by means of $u$ and its derivatives,   there are defined notions  describing various attitudes of the agent towards risk such as: risk aversion, prudence, temperance, etc. (see \cite{arrow}, \cite{eeckhoudt2}, \cite{gollier}, \cite{kimball}, \cite{pratt}).\\

In real life, there are many uncertain situations which are not described by probability theory but by the Zadeh's possibility theory \cite{zadeh}.  Accordingly, it is necessary to develop a possibilistic $EU$-theory built starting from the following three elements:

$\bullet$ a utility function $u$ representing an agent;

$\bullet$ a fuzzy number $A$ representing the risk (with the level sets $[A]^\gamma =[a_1(\gamma), a_2(\gamma)]$, $\gamma \in [0, 1]$);

$\bullet$ a weighting function $f: [0, 1] \rightarrow {\bf{R}}$ ($f$ is a non-negative and increasing function that satisfies $\int_0^1 f(\gamma)d\gamma=1$).\\

Based on three elements , the two following  concepts of possibilistic expected utilities have been introduced (\cite{carlsson1}, \cite{fuller}, \cite{georgescu1}, \cite{georgescu2}):\\

$E_1(f, u(A))=\frac{1}{2} \int_0^1 [u(a_1(\gamma))+u(a_2(\gamma))] f(\gamma)d\gamma,$ (2.1)

$E_2(f, u(A))=\int_0^1 [\frac{1}{a_2(\gamma)-a_1(\gamma)} \int_{a_1(\gamma)}^{a_2(\gamma)} u(x)dx] f(\gamma)d\gamma.$ (2.2)\\

Note that all the integrals which appear in this paper will be assumed finite.\\

Let us specify two particular cases of these two concepts.\\
1) Setting $u=1_{\bf{R}}$ (the identity of $\bf{R}$), the two possibilistic expected utilities introduce the same notion of possibilistic
expected value:

$E_f(A)=E_1(f, 1_{\bf{R}}(A))=E_2(f, 1_{\bf{R}}(A))=\frac{1}{2} \int_{0}^{1} [a_1(\gamma)+a_2(\gamma)]f(\gamma)d\gamma.$ (2.3)\\

2) Setting  $u(x)=(x-E_f(A))^2$ in (2.1) and (2.2), two  possibilistic variances are obtained \footnote{These two types of possibilistic variance have been studied in several papers (see e.g. \cite{carlsson1}, \cite{carlsson2}, \cite{fuller}, \cite{georgescu4}, \cite{zhang}) and have been applied in the study of different possibilistic models (\cite{carlsson2}, \cite{casademunt}, \cite{collan}, \cite{georgescu5}, \cite{mezei}, \cite{majlender}, \cite{thavaneswaran1}, \cite{thavaneswaran2}).}:

$Var_1(f, A)=\frac{1}{2} \int_0^1 [(u(a_1(\gamma))-E_f(A))^2+(u(a_2(\gamma))-E_f(A))^2]f(\gamma) d\gamma,$ (2.4)

$Var_2(f, A)=\int_0^1 [\frac{1}{a_2(\gamma)-a_1(\gamma)} \int_{a_1(\gamma)}^{a_2(\gamma)} (x-E_f(A))^2 dx] f(\gamma) d\gamma.$ (2.5)\\

Let us end this subsection by recalling one type of a fuzzy number useful in this paper (\cite{georgescu4}, Definition 2.3.3):
A triangular fuzzy number $A=(a,\alpha, a+\beta)$ with $a \in {\bf{R}}$ and $\alpha, \beta \geq 0$ is defined by:
$$
A(t) = \left\{
        \begin{array}{ll}
            1-\frac{a-t}{\alpha} & \quad a-\alpha \leq t \leq a, \\
            1- \frac{t-a}{\beta}   & \quad a \leq t \leq a+ \beta,\\
            0.    & \quad otherwise
        \end{array}
    \right.
$$

Then the level sets of $A$ are $[A]^\gamma=[a_1(\gamma), a_2(\gamma)]$, with $a_1(\gamma)=a-(1-\gamma)\alpha$ and $a_2(\gamma)=a+(1-\gamma)\beta$ for any $\gamma \in [0, 1]$, and the support of $A$ is $supp(A)=(a-\alpha, a+\beta).$\\

By Example 3.3.10 from \cite{georgescu4}, the possibilistic expected value $E_f(A)$ associated with the triangular fuzzy number $A=(a,\alpha, \beta)$ has the form:

$E_f(A)=a+\frac{\beta-\alpha}{6}.$ (2.6)\\

In the following subsection, we will recall  the definition of the expected utility operators and a few generalities on them (\cite{georgescu4}, \cite{georgescu5}). We will introduce the $D$-operators by a property regulating the behaviour of an expected utility operator towards the derivation of the utility function with respect to a parameter.\\

We will denote by $\mathcal{F}$ the set of fuzzy numbers and  $\mathcal{C}(\mathbb{R})$ the set of continuous functions from $\mathbb{R}$ to $\mathbb{R}.$  For each $a \in \mathbb{R}$, we denote by $\bar a: \mathbb{R} \rightarrow \mathbb{R}$ the constant
function $\bar a(x)=a$, for $x \in \mathbb{R}.$ $1_{\mathbb{R}}$ will be the identity function of $\mathbb{R}.$ We fix a weighting function $f :[0, 1] \rightarrow \mathbb{R}.$

\subsection{Expected utility operators and D-operators}

Let us recall the expected utility operator.
\begin{definitie}\label{def1}
\cite{georgescu4}, \cite{georgescu5} An {\emph{($f$-weighted) expected utility operator}} is a function $T: \mathcal{F} \times \mathcal{C}(\mathbb{R}) \rightarrow \mathbb{R}$ such that for any
$a, b \in \mathbb{R}$, $g, h \in \mathcal{C}(\mathbb{R})$ and $A \in \mathcal{F}$ the following conditions are fulfilled:

(a) $T(A, 1_{\mathbb{R}})=E_f(A)$;

(b) $T(A, \bar a)=a$;

(c) $T(A, ag+bh)=aT(A, g)+bT(A, h)$;

(d) $g \leq h$ implies $T(A, g) \leq T(A,h)$.

The real number $T(A, g)$ will be called generalized possibilistic expected utility of $A$ w.r.t. $f$ and $g$.
\end{definitie}

Several times in the paper we will write $T(A, g(x))$ instead of $T(A, g)$.

For any integer $k \geq 1$ we define

$\bullet$ the $k$-th order $T$-moment of $A$: $T(A, g)$, where $g(x)=x^k$ for any $x \in {\mathbb{R}}$;

$\bullet$ the $k$-th order central $T$-moment of $A$: $T(A, g)$, where $g(x)=(x-E_f(A))^k$ for any $x \in {\mathbb{R}}$.

In particular, we have the following notions:

$\bullet$ the $T$-variance of $A$: $Var_T(A)=T(A, (x-E_f(A))^2)$;

$\bullet$ the $T$-skewness of $A$: $Sk_T(A)=T(A, (x-E_f(A))^3)$;

$\bullet$ the $T$-kurtosis of $A$: $K_T(A)=T(A, (x-E_f(A))^4)$.\\

The most studied expected utility operators are defined in the following two examples.
\begin{exemplu}\label{exe1}
\cite{georgescu2}, \cite{georgescu3}  The expected utility operator $T_1: \mathcal{F} \times \mathcal{C}(\mathbb{R}) \rightarrow \mathbb{R}$ is defined by:
$$T_1(A, g)=\frac{1}{2} \int_0^1 [g(a_1(\gamma))+g(a_2(\gamma))] f(\gamma)d\gamma.$$
for any fuzzy number $A$ whose level sets are $[A]^\gamma=[a_1(\gamma), a_2(\gamma)]$, $\gamma \in [0, 1]$ and for any $g \in \mathcal{C}(\mathbb{R})$.
\end{exemplu}

\begin{exemplu}\label{exe2}
\cite{georgescu2}, \cite{georgescu3}  The expected utility operator $T_2: \mathcal{F} \times \mathcal{C}(\mathbb{R}) \rightarrow \mathbb{R}$ is defined  by

$T_2(A, g)=\int_0^1 [\frac{1}{a_2(\gamma)-a_1(\gamma)}\int_{a_1(\gamma)}^{a_2(\gamma)} g(x)dx] f(\gamma)d\gamma,$

for any fuzzy number $A$ whose level sets are $[A]^\gamma=[a_1(\gamma), a_2(\gamma)]$, $\gamma \in [0, 1]$ and for any $g \in \mathcal{C}(\bf{R}).$
\end{exemplu}

In the rest of this subsection, we will introduce D-operators useful throughout this paper. For that, we consider the set $\mathcal{U}$ of functions   $g(x, \lambda): \mathbb{R}^2 \rightarrow \mathbb{R}$ satisfying the  following property:\\

 $g(x, \lambda)$ is continuous with respect to $x$ and of class $C^n$  with respect to $\lambda,$\\

%(ii) For any $\lambda \in \mathbb{R}$, the function $\frac{\partial g(., \lambda)}{\partial \lambda}: \mathbb{R} \rightarrow \mathbb{R}$ is continuous.\\

\begin{definitie}\label{def2}
An expected utility operator $T:\mathcal{F} \times \mathcal{U} \rightarrow \mathbb{R}$ is called {\emph{derivable with respect to parameter $\lambda$}} (or, shortly, {\emph{$D$-operator}}) if for any fuzzy number $A$ and for any function $g(x, \lambda) \in \mathcal{U},$  the following conditions are fulfilled:

($D_1$) The function $\lambda \longmapsto T(A, g(., \lambda))$ is derivable (with respect to $\lambda$);

($D_2$) $T(A, \frac{\partial g(., \lambda)}{\partial \lambda})=\frac{d}{d \lambda}T(A, g(., \lambda))$.
\end{definitie}

\begin{propozitie}\label{pro1}
The expected utility operators $T_1$ and $T_2$  are $D$-operators.
\end{propozitie}

\begin{propozitie}\label{pro2}
Let $a_1, a_2 \in \mathbb{R},$  $g,h \in \mathcal{U}.$  and  the function $u=a_1g+a_2h.$ Then:

(a) $u \in \mathcal{U}.$

(b) For any fuzzy number $A,$ the following equality holds:

$\frac{d}{d \lambda} T(A, u(., \lambda))=a_1 \frac{d}{d \lambda} T(A, g(., \lambda))+ a_2 \frac{d}{d \lambda}T(A, h(., \lambda))$.
\end{propozitie}
\begin{proof}
Property (a) is immediate. For (b) we will notice first that

$\frac{\partial u(., \lambda)}{\partial \lambda}=a_1 \frac{\partial g(., \lambda)}{\partial \lambda}+a_2 \frac{\partial h(., \lambda)}{\partial \lambda}$.

Thus, applying Definition \ref{def1} (c) and property ($D_2$) we obtain

$\frac{d}{d \lambda}T(A, u(., \lambda))=T(A, \frac{\partial u(., \lambda)}{\partial \lambda})$

\hspace{2.5cm} $=a_1 T(A, \frac{\partial g(., \lambda)}{\partial \lambda})+a_2 T(A, \frac{\partial h(., \lambda)}{\partial \lambda})$

\hspace{2.5cm}  $=a_1  \frac{d}{d \lambda}T(A, g(., \lambda))+a_2 \frac{d}{d \lambda} T(A, h(., \lambda))$
\end{proof}

In rest of this paper, we will define, in the general framework of a $EU$-theory associated with a $D$-operator and a possibilistic investment model called $T$-standard model. We will then examine optimal solution of the obtained model.
\section{ $T$-standard model in possibilitic portfolio problem}
We will begin the subsection with the short description of the  probabilistic investment model from \cite{eeckhoudt2}, (p. 55-56) which will serve as the starting point in the construction of our $T$-standard model. Then we state our model with its underlying portfolio optimization problem.
\subsection{Portfolio choice problem and $T$-standard model}
 An agent invests an initial wealth $w_0$ in a risk-free asset (bonds) and in a risky asset (stocks). We will assume that the agent has a utility function $u$ of class $\mathcal{C}^2$, increasing and concave. The amount invested in the risky asset will be denoted by $\alpha$, thus $w_0-\alpha$ will be the amount invested in the risk-free asset.

The return of the risky asset is a random variable $X_0$. Let $r$ be the return of the risk-free asset and $x$ a value of $X_0$. The future wealth of the risk-free strategy will be $(w_0-\alpha)(1+r)$, and the value of the portfolio $(w_0-\alpha, \alpha)$ will be $w+\alpha(x-r)$ (by \cite{eeckhoudt2}, p. 65-66) where $w=w_0(1+r).$ Then the investor follows the determination of that $\alpha$ which would be the solution of the optimization problem:

$\displaystyle \max_{\alpha} E(u(w+\alpha(X_0-r)))$ (4.1).

Denoting $X=X_0-r$ the excess return, the problem (4.1) can be written:

$\displaystyle \max_{\alpha} E(u(w+\alpha X))$ (4.2).\\

In monographs \cite{eeckhoudt2}, \cite{gollier} and in papers \cite{garlappi}, \cite{niguez}, \cite{zakamulin} problem (4.2) is studied when the portfolio risk is small, i.e. the excess return $X$ has the form $X=k \mu+Y$ where $\mu >0$ and $Y$ is a random variable with $E(Y)=0$. Then (4.2) gets the form:

$\displaystyle \max_{\alpha} E(u(w+\alpha (\mu k+Y)))$ (4.3)\\
called the probabilistic standard model.\\

In the following, the model (4.3) will inspire us in the choice of the hypotheses and in the formulation of the $T$-standard model.\\

We will fix a weighting function $f: [0, 1] \rightarrow {\bf{R}}$ and a $D$-operator $T: \mathcal{F} \times \mathcal{U} \rightarrow {\bf{R}}$.

At the base of the construction of the $T$-standard model are the following assumptions:

$(H_1)$ The return of the risky asset is a fuzzy number $B_0$;

$(H_2)$ The formulation of the optimization problem will use the notion of generalized possibilistic expected utility associated with the $D$-operator $T$ (see the previous section).

We will denote by $B$ the possibilistic excess return $B_0-r$. We will be in the case of a small possibilistic portfolio risk \footnote{See also \cite{georgescu5}.}, i.e. $B=k\mu+A$, where $\mu >0$ and $A$ is a fuzzy number with $E_f(A)=0$. In this case, $E_f(B)=k \mu.$\\

Getting inspired from (4.3), we will consider the optimization problem:

$\displaystyle \max_{\alpha} T(A, u(w+\alpha (k \mu+x)))$ (4.4)

having the total utility function

$V(\alpha)=T(A, u(w+\alpha (k \mu+x)))$ (4.5).\\

Taking into account condition ($D_2$) from Definition \ref{def2}, we can compute the derivatives:

$V'(\alpha)=\frac{d}{d \alpha} T(A, u(w+\alpha (k \mu+x)))=T(A, (k \mu+x)u'(w+\alpha(k \mu+x)))$ (4.6)

$V''(\alpha)=T(A, (k \mu+x)^2 u''(w+\alpha (k \mu+x)))$  (4.7).

By hypothesis, $u'' \leq 0$, thus, using condition (d) and (b) from Definition \ref{def1}, from (4.7) it will follow $V''(\alpha) \leq 0$. We proved that the function $V(\alpha)$ is concave.\\

Let $\alpha(k)$ be the solution of the optimization problem $\displaystyle \max_{\alpha} V(\alpha)$, where $V(\alpha)$ has the form (4.5). By (4.6), the first-order condition $V'(\alpha(k))=0$ will be written

$T(A, (k \mu+x)u'(w+\alpha(k)(k \mu+x)))=0.$ (4.8)\\

As in the case of the probabilistic model from \cite{gollier},  we assume $\alpha(0)=0$. Assuming that $\alpha(k)$ is of class $C^n$, we will look for Taylor-type approximations of the solution $\alpha(k)$ of the equation (4.8):

$\alpha(k) \approx \displaystyle \sum_{j=0}^{n} \frac{k^j}{j\,!} \alpha^{(j)}(0)$  (4.9)

We approximate the derivative $u'(w+\alpha(k \mu+x))$ by

$u'(w+ \alpha(k \mu+x)) \approx \displaystyle \sum_{j=0}^n \frac{u^{(j+1)}(w)}{j\,!} \alpha^j (k\mu+x)^j$ (4.10)

from where, by multiplying by $k \mu+x$ one obtains:

$(k \mu+x)u'(w+\alpha(k \mu+x)) \approx \displaystyle \sum_{j=0}^n \frac{u^{(j+1)}(w)}{j\,!} \alpha^j (k\mu+x)^{j+1}.$ (4.11)\\

Taking into account condition (c) of Definition \ref{def1}, from (4.11) it follows:

$T(A, (k \mu+x)u'(w+\alpha (k \mu+x))) \approx \displaystyle \sum_{j=0}^n \frac{u^{(j+1)}(w)}{j\,!} \alpha^j T(A, (k \mu+x)^{j+1} )$

By the previous equality, the first-order condition (4.8) gets the form:

$\displaystyle \sum_{j=0}^n \frac{u^{(j+1)}(w)}{j\,!} (\alpha(k))^j T(A, (k \mu+x)^{j+1} )\approx 0.$ (4.12)\\

(4.12) is an $n$-th order equation in the unknown $\alpha(k)$. In most cases, it is difficult to find the exact solution, thus different forms of approximate calculation for $\alpha(k)$ will be searched for.\\

 In the following  subsection, we establish approximate calculation formulas for $\alpha(k)$ in which appear the  indicators of absolute risk aversion and prudence of the utility function $u.$ These indicators are defined as follows:

$r_u(w)=-\frac{u''(w)}{u'(w)}$ (the Arrow-Pratt index of absolute risk aversion \cite{arrow}, \cite{pratt}) (5.1)

$P_u(w)=-\frac{u'''(w)}{u''(w)}$ (the Kimball prudence index \cite{kimball}) (5.2)
\subsection{Optimal allocation based on  absolute risk aversion and prudence}

%In the previous subsection, we have seen that the optimal allocation $\alpha(k)$ of the wealth is the solution of the equation (4.2), and the approximate calculation formula for $\alpha(k)$ can be obtained by solving equation (4.5) for different values of $n$.\footnote{When in a formula the derivative $g^{(n)}$ of a utility function $g$ appears, one will assume that $g$ is of class $C^n$.} We will fix a weighting function $f: [0, 1] \rightarrow {\bf{R}}$, a utility function $u$, a $D$-operator $T$ and a small possibilistic risk $k \mu+A$, where $\mu >0$ and $A$ is a fuzzy number with $E_f(A)=0$. Also w

We will consider the optimization problem (4.5), keeping all the notations from Subsection $3.1.$

We will search for an approximation of the optimal allocation $\alpha(k)$ of the form:

$\alpha(k) \approx \alpha(0)+k \alpha'(0)+\frac{1}{2} k^2 \alpha''(0)=k \alpha'(0)+\frac{1}{2} k^2 \alpha''(0).$ (5.3)\\

The two first key results of this paper determine the approximate values of $\alpha'(0)$ and $\alpha''(0)$ by means of risk aversion and prudence.
\begin{propozitie}\label{pro3}
$\alpha'(0) \approx \frac{\mu}{T(A, x^2)}\frac{1}{r_u(w)}.$
\end{propozitie}
\begin{proof}
For $n=1$, equation (4.12) gets the form

$u'(w) T(A, k \mu+x)+\alpha(k) u''(w)T(A, (k \mu+x)^2) \approx 0$ (5.4)

By Definition \ref{def1}, we have $T(A, k \mu+x)=k \mu+T(A, x)=k \mu+E_f(A)$

Deriving equation (5.4) with respect to $k$ and taking into account $(D_2)$ we will obtain

$u'(w)\mu+u''(w)[\alpha'(k)T(A, (k\mu+x)^2)+2 \alpha(k) \mu T(A, k \mu+x)] \approx 0$

Setting $k=0$ and considering that $\alpha(0)=0$, we have

$u'(w)\mu+u''(w)T(A, x^2)\alpha'(0) \approx 0.$ From where it follows immediately
$\alpha'(0) \approx -\frac{\mu}{T(A, x^2)}\frac{u'(w)}{u''(w)}=\frac{\mu}{T(A, x^2)}\frac{1}{r_u(w)}.$
\end{proof}

\begin{propozitie}\label{pro4}
$\alpha''(0) \approx \frac{P_u(w)}{(r_u(w))^2} \frac{T(A, x^3)}{(T(A, x^2))^3}\mu^2.$
\end{propozitie}
\begin{proof}
For $n=2$, equation (4.12) becomes

$u'(w) T(A, k \mu+x)+u''(w)\alpha(k) T(A, (k \mu+x)^2)+\frac{u'''(w)}{2} (\alpha(k))^2 T(A, (k\mu+x)^3) \approx 0$.

We recall that $T(A, k\mu+x)=k \mu+E_f(A)=k \mu$. We derive the above relation with respect to $k$, considering ($D_2$):

$u'(w)\mu+u''(w)[\alpha'(k)T(A, (k \mu+x)^2+2\alpha(k)\mu T(A, k \mu+x)]+$

$+\frac{u'''(w)}{2} [2\alpha(k) \alpha'(k) T(A, (k\mu+x)^3)+3(\alpha(k))^2 \mu T(A, (k\mu+x)^2)] \approx 0$.

Deriving one more time with respect to $k$, setting $k=0$ and considering $\alpha(0)=0$ and $E_f(A)=0$, we will obtain:

$u''(w)\alpha''(0)T(A, x^2)+u'''(w)(\alpha'(0))^2T(A, x^3) \approx 0$

from where one obtains

$\alpha''(0) \approx - \frac{u'''(w)}{u''(w)} \frac{T(A, x^3)}{T(A, x^2)} (\alpha'(0))^2$.

Taking into account Proposition \ref{pro3}, Eq. (5.1) and Eq. (5.2), it will follow:

$\alpha''(0) \approx \frac{P_u(w)}{(r_u(w))^2} \frac{T(A, x^3)}{(T(A, x^2))^3}\mu^2.$
\end{proof}

The following first main result of this paper establishes the first approximate value of the solution of our $T$-model.
\begin{teorema} \label{teo1}
The optimal allocation $\alpha(k)$ can be approximated as:\\
$\alpha(k) \approx \frac{\mu}{r_u(w)} \frac{1}{T(A, x^2)}k+\frac{1}{2}\mu^2\frac{P_u(w)}{(r_u(w))^2}\frac{T(A, x^3)}{(T(A, x^2))^3}   k^2 $
\end{teorema}
\begin{proof}
$\alpha'(0)$ and $\alpha''(0)$ are replaced in (5.3) with their approximate values from Propositions \ref{pro3} and \ref{pro4}.
\end{proof}

\begin{remarca}\label{rem2}
Since $E_f(A)=0$ it follows that $T(A, x^2)=Var_T(A)$ and $T(A, x^3)=Sk_T(A)$, thus by Theorem \ref{teo1}, the optimal allocation $\alpha(k)$ can be approximated by

$\alpha(k) \approx \frac{ \mu}{r_u(w)} \frac{1}{Var_T(A)}k+\frac{1}{2}  \mu^2 \frac{P_u(w)}{(r_u(w))^2}\frac{Sk_T(A)}{(Var_T(A))^3}k^2.$ (5.5)\\

As seen from (5.5), $\alpha(k)$ is expressed according to:

$\bullet$ the Arrow-Pratt index $r_u(w)$ and the prudence index $P_u(w)$

$\bullet$ the $T$-variance $Var_T(A)$ and the $T$-skewness $Sk_T(A)$.
\end{remarca}

Let us display some examples of our solution.
\begin{exemplu}\label{exe3}
We consider the utility function $u(w)=\frac{w^a}{a}$ with $a \not=0$. A simple computation shows that

$r_u(w)=\frac{1-a}{w}$; $P_u(w)=\frac{2-a}{w}$;
$T_u(w)=\frac{3-a}w$   (5.6)

from which it follows

$\frac{1}{r_u(w)}=\frac{w}{1-a}$; $\frac{P_u(w)}{(r_u(w))^2}=\frac{2-a}{(1-a)^2} w$  (5.7)

Replacing $\frac{1}{r_u(w)}$ and $\frac{P_u(w)}{(r_u(w))^2}$ in the expression of $\alpha(k)$ from Theorem 5.3 we find

$\alpha(k) \approx \frac{k \mu w}{(1-a)T(A, x^2)} [1+\frac{k \mu}{2}\frac{2-a}{1-a}\frac{T(A, x^3)}{(T(A, x^2))^3}]$ (5.8)

We assume that the weighting function is $f(\gamma)=2 \gamma$, $A$ is a triangular fuzzy number $A=(b, \alpha, \beta)$ with $E_f(A)=0$ and $T$ is the $D$-operator $T_1$ from Example 3.2.

By \cite{thavaneswaran2}, Remark  2.1, (a) and (b) we have

$T_1(A, x^2)=Var_{T_1}(A)=\frac{\alpha^2+\beta^2+\alpha \beta}{18}$  (5.9)

$T_1(A, x^3)=Sk_{T_1}(A)=\frac{19(\beta^2-\alpha^2)}{1080}+\frac{\alpha \beta(\beta-\alpha)}{72}$ (5.10)

Then formula (5.8) becomes

$\alpha(k) \approx \frac{18 k \mu w}{(1-a)(\alpha^2+\beta^2 +\alpha \beta)} [1+\frac{9 k \mu(2-a)}{1-a}\frac{\frac{57(\beta^2-\alpha^2)}{10}+\frac{9\alpha \beta (\beta- \alpha)}{2}}{(\alpha^2+\beta^2+\alpha \beta)^3}]$ (5.11)

In case of a symmetric triangular fuzzy number $A=(b, \alpha)$ we have $b=0$, $\alpha=\beta$ and formula (5.11) becomes

$\alpha(k) \approx \frac{6 \mu w}{(1-a)\alpha^2}k.$ (5.12)

\end{exemplu}

\begin{exemplu} \label{exe4}
Assume that the utility function $u$ is HARA-type (see \cite{gollier}, Section 3.6):

$u(w)=\zeta(\delta+\frac{w}{\gamma})^{1-\gamma}$, for $\delta+\frac{w}{\gamma}>0$ (5.13)

According to \cite{gollier}, Section 3.6:

$r_u(w)=(\delta+\frac{w}{\gamma})^{-1}$ and $P_u(w)=\frac{\gamma+1}{\gamma}(\delta+\frac{w}{\gamma})^{-1}$

from where it follows:

$\frac{1}{r_u(w)}=\delta+\frac{w}{\gamma}$ and $\frac{P_u(w)}{(r_u(w))^2}=\frac{\gamma+1}{\gamma}(\delta+\frac{w}{\gamma})$.

Replacing in (5.5) $\frac{1}{r_u(w)}$ and $\frac{P_u(w)}{(r_u(w))^2}$ with the values obtained above, the optimal allocation $\alpha(k)$ will be approximated as:

$\alpha(k) \approx k \mu (\delta+\frac{w}{\gamma})\frac{1}{Var_T(A)}+\frac{1}{2}(k \mu)^2 \frac{\gamma+1}{\gamma}(\delta+\frac{w}{\gamma}) \frac{Sk_T(A)}{(Var_T(A))^3}$  (5.14)

Assume that the weighting function $f$ has the form $f(t)=2t$, $t \in [0, 1]$. The risk $A$ is a triangular fuzzy number $A=(b, \alpha, \beta)$ with $E_f(A)=0$ and $T$ is the $D$-operator $T_1$. By (5.9) and (5.10), in this case formula (5.14) will get the form:

$\alpha(k) \approx 18 \mu(\delta+\frac{w}{\gamma})\frac{1}{\alpha^2+\beta^2+\alpha \beta}k+$

$+\frac{18^3}{2} \mu^2 \frac{\gamma+1}{\gamma}(\delta+\frac{w}{\gamma})\frac{\frac{19(\beta^2-\alpha^2)}{1080}+\frac{\alpha \beta(\beta-\alpha)}{72}}{(\alpha^2+\beta^2+\alpha \beta)^3}k^2.$   (5.15)
\end{exemplu}

In the following subsection, we shall prove an approximate calculation formula for the optimal allocation $\alpha(k)$ in terms of the following parameters:\\

$\bullet$ the indicators of risk aversion, prudence and temperance associated with the utility function $u$;

$\bullet$ $T$-variance, $T$-skewness and $T$-kurtosis associated with the fuzzy number $A$.
\subsection{Optimal allocation in terms of absolute risk aversion, prudence and temperance}

%We recall that we are under the assumption of a small portfolio risk, where the possibilistic excess return $B$ has the form $B=k \mu+A$, with $\mu >0$ and $A$ a fuzzy number with $E_f(A)=0$. Also, we recall the temperance indicator \footnote{As the prudence indicator $P_u(w)=-\frac{u'''(w)}{u''(w)}$, the temperance indicator $T_u(w)$ has been introduced by Kimball in \cite{kimball}, \cite{kimball1}. This expresses the moderation the agent $u$ displays in face of risk.}  associated with $u$:
%
%$T_u(w)=-\frac{u^{iv}(w)}{u'''(w)}$ (6.1)

To find the way the temperance indicator $T_u(w)=-\frac{u^{iv}(w)}{u'''(w)}$ appears in the optimal solution $\alpha(k)$ we will write the formula (4.9) for $n=3$:

$\alpha(k) \approx k \alpha'(0) +\frac{1}{2} k^2 \alpha''(0)+\frac{1}{3\, !}k^3 \alpha'''(0).$ (6.2)\\

The following third key result of this paper establishes  an approximate calculation formula for $\alpha'''(0).$  It emphasizes a dependence relation between $\alpha'(0)$, $\alpha''(0)$ (found in Propositions \ref{pro3} and \ref{pro4}) and $\alpha'''(0).$ The Proof is in Appendix.
\begin{propozitie}\label{pro5}
$\alpha'''(0)T(A, x^2)+6 \alpha'(0)\mu^2 -3P_u(w)[\alpha'(0)\alpha''(0)T(A, x^3)+3 \mu (\alpha'(0))^2T(A, x^2)]+$

$+\frac{T_u(w)}{P_u(w)}(\alpha'(0))^3 T(A, x^4) \approx 0.$
\end{propozitie}

Using the dependence relation of Proposition \ref{pro5}, we will present below a more refined formula of approximate calculation of the optimal solution $\alpha(k).$ Thus, the second main result of this paper establishes the approximate value of the optimal solution of our model.
\begin{teorema}\label{teo2}
$\alpha(k) \approx \frac{k \mu}{r_u(w)}\frac{1}{T(A, x^2)}+\frac{1}{2}(k\mu)^2\frac{P_u(w)}{(r_u(w))^2}\frac{T(A, x^3)}{(T(A, x^2))^3}-$

$-\frac{(k \mu)^3}{r_u(w)}\frac{1}{(T(A, x^2))^2}+\frac{1}{2}(k \mu)^3 \frac{(P_u(w))^2}{(r_u(w))^3}\frac{(T(A, x^3))^2}{(T(A, x^2))^5}+$

$\frac{3}{2} (k\mu)^3 \frac{P_u(w)}{(r_u(w))^2}\frac{1}{(T(A, x^2))^2}-\frac{1}{6} (k \mu)^3 \frac{T_u(w)}{P_u(w)(r_u(w))^3}\frac{T(A, x^4)}{(T(A, x^2))^4}.$
\end{teorema}
\begin{proof}
According to (5.3) and Theorem \ref{teo1},

$k \alpha'(0)+\frac{1}{2}k^2 \alpha''(0) \approx \frac{k \mu}{r_u(w)} \frac{1}{T(A, x^2)}+\frac{1}{2}(k\mu)^2 \frac{P_u(w)}{(r_u(w))^2}\frac{T(A, x^3)}{(T(A, x^2))^3}$  (6.10)

By (6.2), in the component of (6.2), besides the expression (6.10),  $\frac{1}{3\, !} k^3 \alpha'''(0)$ appears. We will compute this term using the  dependence relation between $\alpha'(0)$, $\alpha''(0)$ and $\alpha'''(0)$ from Proposition \ref{pro5}.

We recall the approximate values of $\alpha'(0)$ and $\alpha''(0)$ from Propositions \ref{pro3} and \ref{pro4}:

$\alpha'(0)=\frac{\mu}{r_u(w)} \frac{1}{T(A, x^2)}$; $\alpha''(0)=\mu^2 \frac{P_u(w)}{(r_u(w))^2}\frac{T(A, x^3)}{(T(A, x^2))^3}$ (6.11)

Taking into account (6.11) a simple computation shows that

$\bullet$ $\frac{1}{3\, !} k^3 6 \alpha'(0)\mu^2=\frac{(k \mu)^3}{r_u(w)} \frac{1}{T(A, x^2)}$

$\bullet$ $\frac{1}{3\, !} k^3 3 P_u(w) \alpha'(0) \alpha''(0)T(A, x^3)=\frac{1}{2} (k \mu)^3 \frac{(P_u(w))^2}{(r_u(w))^3}\frac{(T(A, x^3))^2}{T(A, x^2))^4}$

$\bullet$ $\frac{1}{3\, !} k^3 9 P_u(w) \mu(\alpha'(0))^2 T(A, x^2)=\frac{3}{2} (k \mu)^3 \frac{P_u(w)}{(r_u(w))^2}\frac{1}{(T(A, x^2)}$

$\bullet$ $\frac{1}{3\, !} k^3 (\alpha'(0))^3 T(A, x^4) \frac{T_u(w)}{P_u(w)}=\frac{1}{6} (k \mu)^3 \frac{T_u(w)}{P_u(w)(r_u(w))^3} \frac{T(A, x^4)}{(T(A, x^2))^3}$

Multiplying the identity of Proposition \ref{pro5} by $\frac{1}{3\, !} k^3$ and taking into account the four equalities above, it follows:

$\frac{1}{3\, !} k^3 \alpha'''(0)T(A, x^2)+\frac{(k \mu)^3}{r_u(w)}\frac{1}{T(A, x^2)}-$

$-\frac{1}{2}(k \mu)^3 \frac{(P_u(w))^2}{(r_u(w))^3}\frac{(T(A, x^3))^2}{(T(A, x^2))^4}-\frac{3}{2}(k \mu)^3 \frac{P_u(w)}{(r_u(w))^2}\frac{1}{T(A, x^2)}+$

$+\frac{1}{6}  (k \mu)^3 \frac{T_u(w)}{P_u(w)(r_u(w))^3}\frac{T(A, x^4)}{(T(A, x^2))^3}\approx 0.$

From this equation we find the value of $\frac{1}{3\, !} k^3 \alpha'''(0)$:

$\frac{1}{3\, !} k^3 \alpha'''(0) \approx -\frac{(k \mu)^3}{r_u(w)}\frac{1}{(T(A, x^2))^2}+\frac{1}{2}(k \mu)^3 \frac{(P_u(w))^2}{(r_u(w))^3}\frac{(T(A, x^3))^2}{(T(A, x^2))^5}+$

$+\frac{3}{2} (k \mu)^3 \frac{P_u(w)}{(r_u(w))^2}\frac{1}{(T(A, x^2))^2}-\frac{1}{6}(k \mu)^3 \frac{T_u(w)}{P_u(w)(r_u(w))^3}\frac{T(A, x^4)}{(T(A, x^2))^4}.$

Replacing in (6.2), $k \alpha'(0)+\frac{1}{2}k^2 \alpha''(0)$ with the value from (6.10) and $\frac{1}{3\, !} k^3 \alpha'''(0)$ with the above computed value it follows for $\alpha(k)$ the approximate value from the enunciation.
\end{proof}

We rewrite our second main result by means of the four first order-central moments of the risky asset.
\begin{corolar} \label{cor1}
$\alpha(k) \approx \frac{ \mu}{r_u(w)}\frac{1}{Var_T(A)}k+\frac{1}{2}( \mu)^2 \frac{P_u(w)}{(r_u(w))^2}\frac{Sk_T(A)}{(Var_T(A))^3}k^2+$

$(-\frac{\mu^3}{r_u(w)} \frac{1}{(Var_T(A))^2}+\frac{1}{2} \mu^3 \frac{(P_u(w))^2}{(r_u(w))^3}\frac{(Sk_T(A))^2}{(Var_T(A))^5}+$

$+\frac{3}{2} \mu^3 \frac{P_u(w)}{(r_u(w))^2}\frac{1}{(Var_T(A))^2}-\frac{1}{6} \mu^3 \frac{T_u(w)}{P_u(w)(r_u(w))^3}\frac{K_T(A)}{(Var_T(A))^4})k^3.$
\end{corolar}

\begin{proof}
Since $E_f(A)=0$ we will have $T(A, x^2)=Var_T(A)$, $T(A, x^3)=Sk_T(A)$ and $T(A, x^4)=K_T(A)$.
\end{proof}

The approximate expression of $\alpha(k)$ from Corollary \ref{cor1} is quite complicated. Therefore, denoting

$F_1=\frac{1}{r_u(w)} \frac{1}{Var_T(A)}$; $F_2=\frac{P_u(w)}{(r_u(w))^2} \frac{Sk_T(A)}{(Var_T(A))^3}$;

$F_3=\frac{1}{r_u(w)} \frac{1}{(Var_T(A))^2}$; $F_4=\frac{(P_u(w))^2}{(r_u(w))^3} \frac{(Sk_T(A))^2}{(Var_T(A))^5}$;  (6.12)

$F_5=\frac{P_u(w)}{(r_u(w))^2} \frac{1}{(Var_T(A))^2}$;

 $F_6=\frac{T_u(w)}{P_u(w)(r_u(w))^3} \frac{K_T(A)}{(Var_T(A))^4}.$\\

From Corollary \ref{cor1} we will obtain:
\begin{remarca} \label{rem2}
The approximate value of the optimal allocation $\alpha(k)$ will be computed with the following formula:

$$\alpha(k) \approx k \mu F_1+ \frac{1}{2} (k \mu)^2 F_2-(k \mu)^3 [F_3-\frac{1}{2}F_4-\frac{3}{2}F_5+\frac{1}{6}F_6]$$  (6.13)
\end{remarca}

To obtain the approximate value of $\alpha(k)$ we will compute first: $F_1$, \ldots, $F_6$ with formulas (6.12), then these will be replaced in (6.13).
\begin{exemplu}\label{exe5}
We consider the possibilistic portfolio problem (4.4) with the initial data:

$\bullet$ the weighting function is $f(t)=2t$, $t \in [0, 1]$

$\bullet$ the agent's utility function is $u(w)=w^a$, $a >0$.

Then, by Example 5.5

$r_u(w)=\frac{1-a}{w}$; $P_u(w)=\frac{2-a}{w}$; $T_u(w)=\frac{3-a}{w}$;

$\frac{1}{r_u(w)}=\frac{w}{1-a}$; $\frac{P_u(w)}{(r_u(w))^2}=\frac{2-a}{(1-a)^2}w$.

Moreover, we will have:

$\frac{(P_u(w))^2}{(r_u(w))^3}=\frac{(2-a)^2}{(1-a)^3}w$; $\frac{T_u(w)}{P_u(w)(r_u(w))^3}=\frac{3-a}{(2-a)(1-a)^3}w^3$.

Replacing in (6.12), it follows

$F_1=\frac{w}{1-a}\frac{1}{Var_T(A)}$; $F_2=\frac{(2-a)w}{(1-a)^2}\frac{Sk_T(A)}{(Var_T(A))^3}$; $F_3=\frac{w}{1-a}\frac{1}{(Var_T(A))^3}$;

$F_4=\frac{(2-a)^2w}{(1-a)^3}\frac{(S_k T(A))^2}{(Var_T(A))^5}$; $F_5=\frac{(2-a)w}{(1-a)^2}\frac{1}{(Var_T(A))^2}$; $F_6=\frac{(3-a)w^3}{(2-a)(1-a)^3}\frac{K_T(A)}{(Var_T(A))^4}$.

Assume that the risk $A$ is represented by the triangular fuzzy number $A=(b, \alpha, \beta)$ and $T$ is the $D$-operator $T_1$. Then $Var_{T_1}(A)$ and $Sk_{T_1}(A)$ can be computed with formulas (5.9) and (5.10). By \cite{thavaneswaran2}, Remark 2.1 (2), for $K_{T_1}(A)$ we have the following value:

$T(A, x^4)=K_{T_1}(A)=\frac{\beta^2 \alpha^2}{72}+\frac{5(\alpha^4+\beta^4)}{432}+\frac{2 \alpha \beta(\alpha^2+\beta^2)}{135}$

Replacing the values of $Var_{T_1}(A)$, $Sk_{T_1}(A)$ and $K_{T_1}(A)$ in the above expressions of $F_1$-$F_6$, we find the following forms of them:

$F_1=\frac{w}{1-a}\frac{18}{\alpha^2+\beta^2+\alpha \beta},$

$F_2=\frac{(2-a)w}{(1-a)^2}\frac{18^3[\frac{19(\beta^2-\alpha^2)}{1080}+\frac{\alpha \beta (\beta-\alpha)}{72}]}{(\alpha^2+\beta^2+\alpha \beta)^3},$

$F_3=\frac{w}{1-a}\frac{18^3}{(\alpha^2+\beta^2+\alpha \beta)^3},$

$F_4=\frac{(2-a)^2w}{(1-a)^3}\frac{18^5[\frac{19(\beta^2-\alpha^2)}{1080}+\frac{\alpha \beta (\beta-\alpha)}{72}]^2}{(\alpha^2+\beta^2+\alpha \beta)^5},$

$F_5=\frac{(2-a)w}{(1-a)^2}\frac{324}{(\alpha^2+\beta^2+\alpha \beta)^2},$

$F_6=\frac{(3-a)w^3}{(2-a)(1-a)^3} \frac{1458\alpha^2\beta^2+1215(\alpha^4+\beta^4)+\frac{2 \times 18^4}{135}\alpha \beta (\alpha^2+\beta^2)}{(\alpha^2+\beta^2+\alpha \beta)^4}.$

Replacing the obtained values of $F_1$, \ldots, $F_6$ in (6.13), one obtains the approximate value of the optimal allocation $\alpha(k)$.

\end{exemplu}
\section{Concluding Remarks}

%The expected utility operators have been introduced in \cite{georgescu3} to develop a general theory of possibilistic risk aversion. Each expected utility operator acts on the fuzzy number and the utility function, with the preservation of the linearity of the latter. An expected utility operator defines a
%notion of expected utility, and by this, a possibilistic $EU$-theory. Still this context is too general for the treatment of themes from possibilistic risk modeling. One of this themes is the approximation of the solution of possibilistic portfolio problems.

In this paper, we adress  the  optimization portfolio problem in the framework of a possibilistic $EU$-theory when risky asset is a fuzzy number.

The first contribution of the paper is the introduction of special expected utility operators, called $D$-operators. These are defined by preserving the partial derivability of the utility function with respect to a parameter, which will allow the study of the first order conditions of the optimization problems.

The second contribution of the paper is the formulation of a possibilistic portfolio choice problem inside of $EU$-theory associated with a $D$-operator $T$.

The third contribution is the proof of an approximate formula for the solution of the optimization problem associated with that portfolio problem based on  the indicators of the investor's preferences (risk aversion, prudence, temperance) and the possibilistic moments associated with $T$.

An open problem is that in the context of the $D$-operator to study models with two types of risk: besides the investment risk from the standard model to appear a background risk. Both the investment risk and the background risk can be either random variables or fuzzy numbers, thus we would have four models per total. For each four background risk models we should find approximations of the corresponding optimization problems, such that in the particular case $T=T_1$ to be found the results from \cite{georgescu5}.

In paper \cite{kaluszka} the Jensen-type operators have been defined, a notion which considerably extends the expected utility operators. They can not only act on fuzzy numbers, but on random fuzzy numbers, type-2 fuzzy sets, random type-2 fuzzy sets, etc.

Even if the Jensen-type operators are not linear, they allow the development of new risk aversion theories. In particular, in \cite{kaluszka}, a very general Arrow-Pratt type theorem is proved. Introducing for Jensen-type operators the linearity condition and some axioms similar to ($D_1$), ($D_2$), we could obtain a notion of $D$-Jensen-type operators which should extend our $D$-operators. An open problem would be the generalization of the results from this paper for $D$-Jensen-type operators.

\subsection*{Appendix: Proof of Proposition \ref{pro5}  }
\begin{proof}
For $n=3$, the first-order condition (4.12) gets the form

$u'(w) T(A, k \mu+x)+u''(w)\alpha(k) T(A, (k \mu+x)^2)+\frac{u'''(w)}{2} (\alpha(k))^2 T(A, (k \mu+x)^3)+\frac{u^{iv}(w)}{3\, !} (\alpha(k))^3T(A, (k \mu+x)^4) \approx 0$.

From the previous section we know that $T(A, k \mu+x)=k \mu$. For a better structure of the computations we will denote

$T_1(k)=\alpha(k) T(A, (k \mu+x)^2)$ (6.3)

$T_2(k)=(\alpha(k))^2 T(A, (k \mu+x)^3)$ (6.4)

$T_3(k)=(\alpha(k))^3 T(A, (k \mu+x)^4)$ (6.5)

With these notations, the above equation will be written:

$u'(w) k \mu+u''(w)T_1(k)+\frac{u'''(w)}{2}T_2(k)+\frac{u^{iv}(w)}{3\, !}T_3(k) \approx 0$ (6.6)

from where, deriving three times, one will obtain:

$u''(w)T_1'''(k)+\frac{u'''(w)}{2}T_2'''(k)+\frac{u^{iv}(w)}{3\, !}T_3'''(k) \approx 0$ (6.7)

For $k=0$, it follows

$u''(w)T_1'''(0)+\frac{u'''(w)}{2}T_2'''(0)+\frac{u^{iv}(w)}{3\, !}T_3'''(0) \approx 0$ (6.8)

We recall  from calculus that for the three times derivable functions $f$, $g$, we have

$(fg)'''=f'''g+3f''g'+3f'g''+fg'''$ (6.9)

Formula (6.9) will be used to determine $T_1'''(0)$, $T_2'''(0)$ and $T_3'''(0)$.

{\emph{The computation of $T_1'''(0)$}}

 By (6.9) we have

$T_1'''(k)=\alpha'''(k)T(A, (k\mu+x)^2)+3 \alpha''(k)\frac{d}{dk}T(A, (k \mu+x)^2)+$

$+3 \alpha'(k)\frac{d^2}{dk^2}T(A, (k \mu+x)^2)+\alpha(k) \frac{d^3}{dk^3}T(A, (k \mu+x)^2)$

Applying condition $(D_2)$ we will have

$\frac{d}{dk}T(A, (k \mu+x)^2)=2 \mu T(A, k\mu+x)=2 \mu^2 k$; $\frac{d^2}{dk^2}T(A, (k \mu+x)^2)=2\mu^2$

from where it follows

$\frac{d}{dk}T(A, (k \mu+x)^2)|_{k=0}=0$; $\frac{d^2}{dk^2}T(A, (k \mu+x)^2)|_{k=0}=2 \mu^2$.

Replacing these values in the above expression of $T_1'''(k)$ and knowing that $\alpha(0)=0$, it will follow

$T_1'''(0)=\alpha'''(0)T(A, x^2)+6 \alpha'(0)\mu^2$

{\emph{The computation of $T_2'''(0)$}}

If we denote $g(k)=(\alpha(k))^2$, then, knowing (6.9), $T_2'''(k)$ is written

$T_2'''(k)=g'''(k)T(A, (k\mu+x)^3)+3 g''(k)\frac{d}{dk}T(A, (k \mu+x)^3)+$

$+3g'(k)\frac{d^2}{dk^2}T(A, (k \mu+x)^3)+g(k)\frac{d^3}{dk^3}T(A, (k \mu+x)^3)$

To determine $T_2'''(0)$ we will compute the values of all terms in the component of $T_2'''(k)$ for $k=0$.

We remark that

$g(k)=(\alpha(k))^2$; $g(0)=0$;

$g'(k)=2\alpha'(k)\alpha(k)$; $g'(0)=0$;

$g''(k)=2[\alpha''(k)\alpha(k)+(\alpha'(k))^2]$; $g''(0)=2(\alpha'(0))^2$;

$g'''(k)=2[\alpha'''(k)\alpha(k)+3\alpha'(k)\alpha''(k)]$; $g'''(0)=6 \alpha'(0)\alpha''(0)$.

Applying condition $(D_2)$ we will compute the new derivatives:

$\frac{d}{dk}T(A, (k \mu+x)^3)=3 \mu T(A, (k \mu+x)^2)$; $\frac{d}{dk}T(A, (k \mu+x)^3)|_{k=0}=3 \mu T(A, x^2)$;

$\frac{d^2}{dk^2}T(A, (k \mu+x)^3)=6 \mu^2 T(A, k\mu+x)=6\mu^3k$.

With these computations we obtain the expression of $T_2'''(0)$:

$T_2'''(0)=g'''(0)T(A, x^3)+3g''(0)3\mu T(A, x^2)$

\hspace{1cm} $=6 \alpha'(0)\alpha''(0)T(A, x^3)+18 \mu (\alpha'(0))^2T(A, x^2)$

{\emph{The computation of $T_3'''(0)$}}

We denote $h(k)=(\alpha(k))^3$. By (6.5) and (6.9) we have

$T_3'''(k)=h'''(k)T(A, (k\mu+x)^4)+3h''(k)\frac{d}{dk}T(A, (k \mu+x)^4)+$

$+3h'(k)\frac{d^2}{dk^2}T(A, (k \mu+x)^4)+h(k)\frac{d^3}{dk^3}T(A, (k \mu+x)^4)$.

By a simple computation one obtains

$h(0)=h'(0)=h''(0)=0$; $h'''(0)=6(\alpha'(0))^3$

from where, setting $k=0$ in the above expression of $T_3'''(k)$, it follows

$T_3'''(k)=h'''(0)T(A, x^4)=(6\alpha'(0))^3T(A, x^4)$

Replacing the found values of $T_1'''(0)$, $T_2'''(0)$, $T_3'''(0)$ in the equation (6.8), one obtains

$u''(w)[\alpha'''(0)T(A, x^2)+6\alpha'(0)\mu^2]+$

$+\frac{u'''(w)}{2}[6 \alpha'(0)\alpha''(0)T(A, x^3)+18 \mu (\alpha'(0))^2T(A, x^2)]+$

$+\frac{u^{iv}(w)}{6}[6 (\alpha'(0))^3T(A, x^4)]\approx 0.$

Dividing by $u''(w)$ and taking into account that $\frac{u^{iv}(w)}{u''(w)}=\frac{T_u(w)}{P_u(w)}$ it follows the equation from
Proposition \ref{pro5}.
\end{proof}

\end{document}